\documentclass[aps,prapplied,reprint, superscriptaddress]{revtex4-2}

\usepackage{amssymb}
\usepackage{amsmath}
\usepackage{amsthm}
\usepackage{qip}
\usepackage{braket}
\usepackage{mathtools}
\usepackage{slashbox}
\usepackage[colorlinks=true,citecolor=blue,urlcolor=black]{hyperref}
\usepackage{array}
\usepackage{pgf}

\newtheorem{definition}{Definition}[section]
\newtheorem{theorem}{Claim}
\newtheorem{lemma}{Lemma}

\newtheorem{corollary}{Corollary}[theorem]

\newcommand*\diff{\mathop{}\!\mathrm{d}}

\DeclarePairedDelimiter\abs{\lvert}{\rvert}%
\DeclareMathOperator{\Tr}{Tr}
\DeclareMathOperator{\vac}{\text{vac}}
\newcolumntype{L}{>{$}l<{$}} % math-mode version of "l" column type

\begin{document}

%\preprint{}

\title{Improved coherent one-way quantum key distribution for high-loss channels}

\author{Emilien \surname{Lavie}}
\email[]{emilien.lavie@u.nus.edu}
\affiliation{Department of Electrical \& Computer Engineering, National University of Singapore, Singapore }

\author{Charles C.-W. \surname{Lim}}
\affiliation{Department of Electrical \& Computer Engineering, National University of Singapore, Singapore }

\date{\today}

\begin{abstract}

The coherent one-way (COW) quantum key distribution (QKD) is a highly practical quantum communication protocol that is currently deployed in off-the-shelves products.
However, despite its simplicity and widespread use, the security of COW-QKD is still an open problem.
This is largely due to its unique security feature based on inter-signal phase distribution, which makes it very difficult to analyze using standard security proof techniques.
Here, to overcome this problem, we present a simple variant of COW-QKD and prove its security in the infinite-key limit.
The proposed modifications only involve an additional vacuum tail signal following every encoded signal and a balanced beam-splitter for passive measurement basis choice.
Remarkably, the resulting key rate of our protocol is comparable with both the existing upper-bound on COW-QKD key rate and the secure key rate of the coherent-state BB84 protocol.
Our findings therefore suggest that the secured deployment of COW-QKD systems in high loss optical networks is indeed feasible with minimal adaptations applied to its hardware and software.

\end{abstract}

\maketitle

\section{Introduction}
Quantum key distribution (QKD) is a promising application of quantum communications where two users, Alice and Bob, exchange quantum signals to establish a common secret key~\cite{gisin_quantum_2002,scarani_security_2009}.
The original ideas of QKD were first presented using the transmission of single photon states \cite{bennett_quantum_2014}, but the field has since evolved to include more practical communication systems based on coherent states. One prime example is the coherent one-way (COW)-QKD protocol \cite{stucki_fast_2005}, which uses a sequence of randomly modulated coherent states with fixed reference phase to distribute the secret key.
In this protocol, each secret bit is encoded into the time-of-arrival of a single light pulse and security is evaluated by checking the optical coherence of consecutive light pulses. The basic idea is to check if the optical coherence between consecutive light pulses has been disturbed---indeed, if an eavesdropper tries to measure the position of light pulse and learn the secret bit, the optical coherence between adjacent non-vacuum light pulses will be broken. This security feature was originally designed to deter the so-called photon-number-splitting (PNS) attack, which raised serious security concerns when it was first discovered in 2001~\cite{lutkenhaus_quantum_2002}.

However, the idea of using inter-signal correlation to detect PNS attacks creates a new problem. In particular, it puts the protocol in an unorthodox situation involving the analysis of sequential trains of pulses, which is fundamentally different from the standard setting based on repeated rounds of quantum communications~\cite{scarani_security_2009}. Consequently, none of the QKD proof techniques developed to date can be applied to COW-QKD. In fact, at the moment, the general security of COW-QKD remains an open problem.

That said, significant progress has been made towards demystifying the security of COW-QKD. Initially, upper bounds on the achievable secret key rate were derived using specific class of \emph{collective} attacks, which suggested a secret key rate that is linear in the channel transmittance~\cite{branciard_upper_2008, korzh_provably_2015}. However, recent studies found tighter upper bounds which feature quadratic scaling~\cite{gonzalez-payo_upper_2020, trenyi_zero-error_2021}. These results are significant because they indicate that COW-QKD may not be suitable for ultra-long-distance QKD, which are similar to what Ref~\cite{branciard_zero-error_2007} have found based on unambiguous state discrimination (USD) attacks assuming zero-error statistics. While these upper bounds provide a clear idea of what COW-QKD could theoretically achieve with lossy channels, it is not obvious if lower bounds with quadratic scaling could be obtained. We note that certain variants of COW-QKD have achieved quadratic scaling with more sophisticated optical receivers based on active switching~\cite{moroder_security_2012, wang_characterising_2019}, but these increase the complexity of the implementation. To better illustrate the current security status of COW-QKD, we plotted some of the known upper and lower bounds in Fig.~\ref{fig:keyrate_summary} assuming zero error statistics.
\begin{figure*}
  \resizebox{\textwidth}{!}{\input{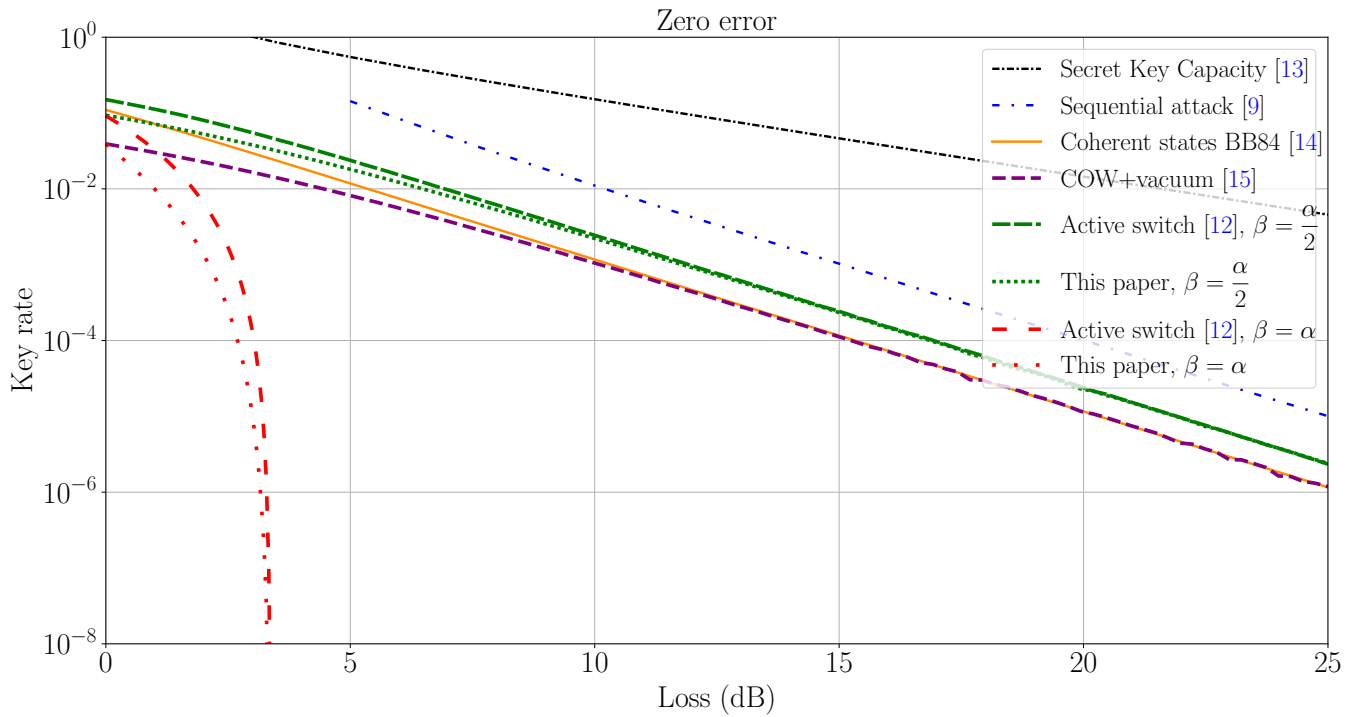}}
  \caption{We compare existing results on the asymptotic security of coherent one-way type protocols.
  For the sake of comparison, we indicate another popular type of protocol based on phase-encoded coherent states \cite{lo_security_2007}.
  Ref.~\cite{pirandola_fundamental_2017} is a fundamental upper limit for point-to-point communication.
  Ref.~\cite{trenyi_zero-error_2021} is a recent result providing an upper bound that scales only quadratically with the channel transmittance.
  Ref.~\cite{wang_characterising_2019} provides a lower bound on the secret key rate for a modified version of the protocol using an active switch instead of a passive interferometer and an optimised intensity for the test state $\ket{\beta}\ket{\beta}$.
  The protocol in this paper is using a passive interferometer for Bob as in the original design, and prepared states similar to Ref.~\cite{wang_characterising_2019} with an extra vacuum pulse sent by Alice. See more details in Section \ref{sec:method}.
  We also analyse the performance of one of the countermeasures proposed in Ref.~\cite{curty_foiling_2021} using a fourth state composed of vacuum pulses only.
  }
  \label{fig:keyrate_summary}
\end{figure*}

In this work, we show that COW-QKD can reach quadratic scaling---close to the bound established in Ref.~\cite{trenyi_zero-error_2021} with only a slight modification of the original protocol. In particular, the proposed protocol is the same as the original protocol, except for (1) an additional vacuum tail signal that is needed to keep the protocol in the standard (iid) setting and (2) a balanced beam-splitter is used to decide passively the measurement basis. To analyse the security of the protocol which is based on practical photon-counting detectors, we use a generalised form of \emph{universal squashing} to map the Hilbert space of the detectors to a two-dimensional Hilbert space~\cite{fung_universal_2011}. Then, we calculate the achievable secret key rate using the standard Shor-Preskill formalism for qubit channels by optimising the phase error rate given the expected channel statistics~\cite{moroder_security_2012}.

The rest of the paper is organised as follows.
In Section~\ref{sec:method} we first provide a detailed model of the protocol implementation (\ref{ssec:modeling}).
Then, in Section~\ref{ssec:squashing} we show how to use the universal squashing framework to estimate the statistics of a virtual single-photon protocol based on the expected statistics of the actual protocol.
We conclude the security analysis in Section~\ref{ssec:winick} where we use a numerical method to estimate a lower bound on the secure key rate of the single-photon protocol.
Finally, we present simulated results and discuss their relevance in Section~\ref{sec:results}.

\section{Method}
\label{sec:method}
\subsection{Modeling}
\label{ssec:modeling}
The protocol we consider here is based on the preparation and measurement of coherent states in three consecutive temporal modes labeled $c_0$, $c_1$ and $c_2$.
The global phase information of the laser used to prepare the states is public and known to the adversary. Any other degree of freedom is assumed to be random and do not carry any useful information about the random inputs. Here, the transmitter, Alice, prepares and sends the state $\ket{\varphi_i}$ with probability $p_i$ chosen from a predefined set:
\begin{equation}
  \begin{cases}
    \ket{\varphi_0} = \ket{\alpha}_{c_0}\ket{\vac}_{c_1}\ket{\vac}_{c_2},\\
    \ket{\varphi_1} = \ket{\vac}_{c_0}\ket{\alpha}_{c_1}\ket{\vac}_{c_2},\\
    \ket{\varphi_2} = \ket{\beta}_{c_0}\ket{\beta}_{c_1}\ket{\vac}_{c_2}.
  \end{cases}
\end{equation}
Note that Alice always sets the third temporal mode $c_2$ to the vacuum state and this is needed to ensure the protocol can be treated in the standard quantum communication setting.

The receiver, Bob, performs decoding measurement using a passive basis choice setup, which is implemented using a balanced beam-splitter leading to two possible detection lines.
The first line is a direct time-of-arrival detection line (key basis labeled $\cZ$): a threshold detector measures the presence or absence of photons in each temporal mode.
The second line is a monitoring line (test basis labeled $\cX$): a Mach-Zender interferometer that interferes consecutive pulses to check for good optical coherence.

Bob's monitoring line is such that only the middle temporal mode $c_1$ arriving on the detectors will contain relevant (conclusive) information.
The first temporal mode contains only one half of the first pulse sent by Alice that did not interfere with anything (since there is no residual light in the interferometer initially), and the third temporal mode contains half of the middle pulse sent by Alice.
Therefore, the conclusive rate in the monitoring line is only half of the one in the direct line. The overall setup is presented in Fig.~\ref{fig:cow-setup}.

\begin{figure}
  \includegraphics[width=\columnwidth]{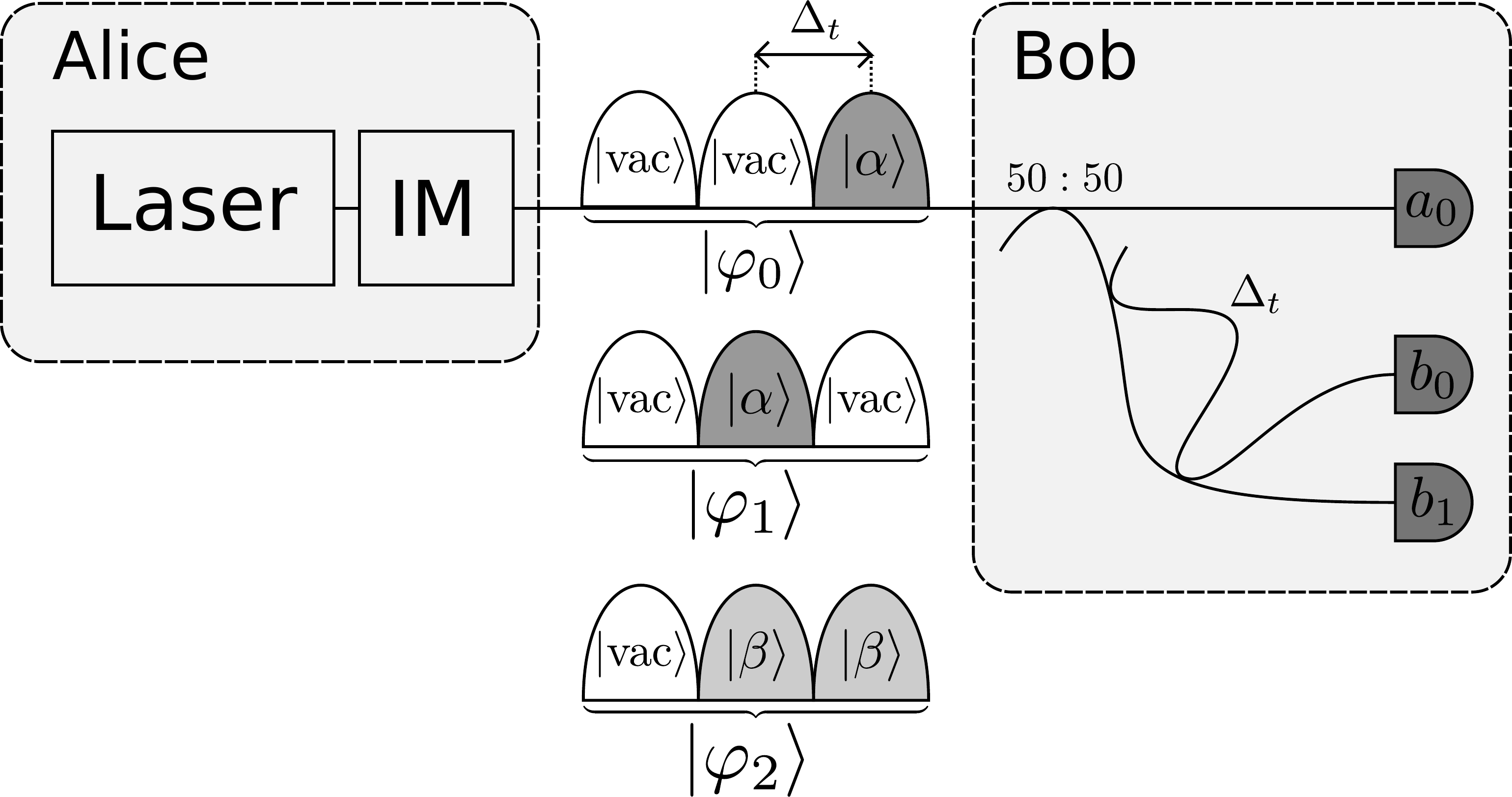}
  \caption{Protocol setup. Alice is preparing three different sequences of coherent states in temporal modes $c_0$, $c_1$, $c_2$: the sequence $\ket{\vac}_{c_2}\ket{\vac}_{c_1}\ket{\alpha}_{c_0}$ is representing bit $0$, the sequence $\ket{\vac}_{c_2}\ket{\alpha}_{c_1}\ket{\vac}_{c_0}$ is representing bit $1$ and the sequence $\ket{\vac}_{c_2}\ket{\beta}_{c_1}\ket{\beta}_{c_0}$ is used to test the channel.
  The last temporal mode $c_2$ is always set to vacuum by Alice to ensure the symmetry of the protocol. Bob is using a beamsplitter to passively choose between a direct line and a monitoring line composed of a Mach-Zender interferometer. Here, the coherence is only monitored between temporal modes $c_0$ and $c_1$ within the same train of pulses, and any coherence with another train is ignored. This is to guarantee a symmetry between the rounds of the protocol and ensure that an optimal collective attack is also an optimal coherent attack \cite{renner_symmetry_2007}. There are two minor differences with the original setup presented in Ref.~\cite{stucki_fast_2005}: an additional vacuum state is enforced at the end of each train, and we allow a lower intensity $\beta \leq \alpha$ for the test state $\ket{\varphi_2}$. }
  \label{fig:cow-setup}
\end{figure}

There are a two minor differences from the original COW protocol \cite{stucki_fast_2005}.
First, the use of a vacuum state at the end of each sequence breaks the coherence between two trains of pulses.
Therefore it is only possible to monitor the coherence within a single train of pulses.
In the original setup, the coherence between any two non-empty subsequent pulses is monitored.
This additional vacuum state makes the security analysis simpler since now an optimal collective attack is also an optimal coherent attack by virtue of symmetries \cite{renner_symmetry_2007}.
Second, we allow the possibility to use a lower intensity for the test sequence $\ket{\varphi_2}$ similar to Ref.~\cite{wang_characterising_2019}.
We expect that these modifications only require minor changes to the original design: only an additional intensity level is used at the transmitter and no additional phase modulator is required.

Now we provide a more detailed notation to describe the transmitter and the receiver.
We label the spatial mode corresponding to the only detector arm in the direct line $a_0$ and the two spatial modes corresponding to the two detector arms in the monitoring line $b_0, b_1$ (see in Fig.~\ref{fig:cow-setup}).

Additionally, we define $9$ modes corresponding to the combination of the spatial and temporal modes.
\begin{equation}
  \begin{matrix}
    d_0 = (a_0, c_0) & d_3 = (a_0, c_1)  & d_6 = (a_0, c_2) \\
    d_1 = (b_0, c_0) & d_4 = (b_0, c_1)  & d_7 = (b_0, c_2) \\
    d_2 = (b_1, c_0) & d_5 = (b_1, c_1)  & d_8 = (b_1, c_2)
  \end{matrix}
\end{equation}
We use here the notation $(a_0, c_0)$ to represent the spatial-temporal mode of the direct line $a_0$ during the first time detection window $c_0$ and similarly for the others.
We use this notation to avoid confusion later when we use the creation/annihilation operators for the modes.
Indeed, populating one photon in mode $d_0$ is effectively considering $d_0^{\dagger}\ket{\vac} = \ket{1}_{d_0}$ which is very different from the product $a_0^{\dagger}c_0^{\dagger}\ket{\vac}_{a_0}\ket{\vac}_{c_0} = \ket{1}_{a_0}\ket{1}_{c_0}$ (one photon in each mode).

Alice only has access to the spatial mode $a_0$ which is the physical channel (e.g. optical fiber or free space) over which she is sending the train of coherent states.
The two spatial modes $b_0, b_1$ come from the empty ports of Bob's beamsplitters: one for the basis choice, one for the first beamsplitter in the interferometer.
Alice (and Eve) have no access to these and they can only affect modes $d_0, d_3, d_6$ instead.
We rephrase Alice's prepared states:

\begin{equation}
  \begin{cases}
    \ket{\varphi_0} = \ket{\alpha}_{d_0}\ket{\vac}_{d_3}\ket{\vac}_{d_6}\\
    \ket{\varphi_1} = \ket{\vac}_{d_0}\ket{\alpha}_{d_3}\ket{\vac}_{d_6}\\
    \ket{\varphi_2} = \ket{\beta}_{d_0}\ket{\beta}_{d_3}\ket{\vac}_{d_6}
  \end{cases}
\end{equation}
We assume all the other unspecified modes are trusted and populated with vacuum states.

Let us now analyse Bob's measurement apparatus.
The beamsplitters are operating on two spatial modes at each time-window.
We find for the basis choice beamsplitter:
\begin{equation}
  \forall c,
  \begin{pmatrix}
    (a_0, c) \\
    (b_0, c) \\
  \end{pmatrix}
  \longrightarrow
  \begin{pmatrix}
    (\frac{a_0 + b_0}{\sqrt{2}}, c) \\
    (\frac{a_0 - b_0}{\sqrt{2}}, c) \\
  \end{pmatrix}
\end{equation}

and similarly for the others.

The delay line can be seen as a shift between the temporal modes: the first temporal mode becomes the second, the second temporal mode becomes the third.
The third temporal mode is actually never populated by Alice or Bob and it is simply an artifact of the computation.
We can always assume that the delay line is effectively transforming it into the first temporal mode.
All in all we have for the only spatial mode that is delayed:
\begin{equation}
  \begin{pmatrix}
    (b_0, c_0) \\
    (b_0, c_1) \\
    (b_0, c_2)
  \end{pmatrix}
  \longrightarrow
  \begin{pmatrix}
    (b_0, c_1) \\
    (b_0, c_2) \\
    (b_0, c_0)
  \end{pmatrix}
\end{equation}
The other modes remain unchanged.

We represent the overall transformation performed by Bob on the input state (before the detectors) with the circuit in Fig.~\ref{fig:circuit-timebin}.
It is easy to check that this circuit is representing a unitary transformation $U^{\dagger}$, since its inverse is given by the reverse circuit.

\begin{figure}
  \includegraphics[width=\columnwidth]{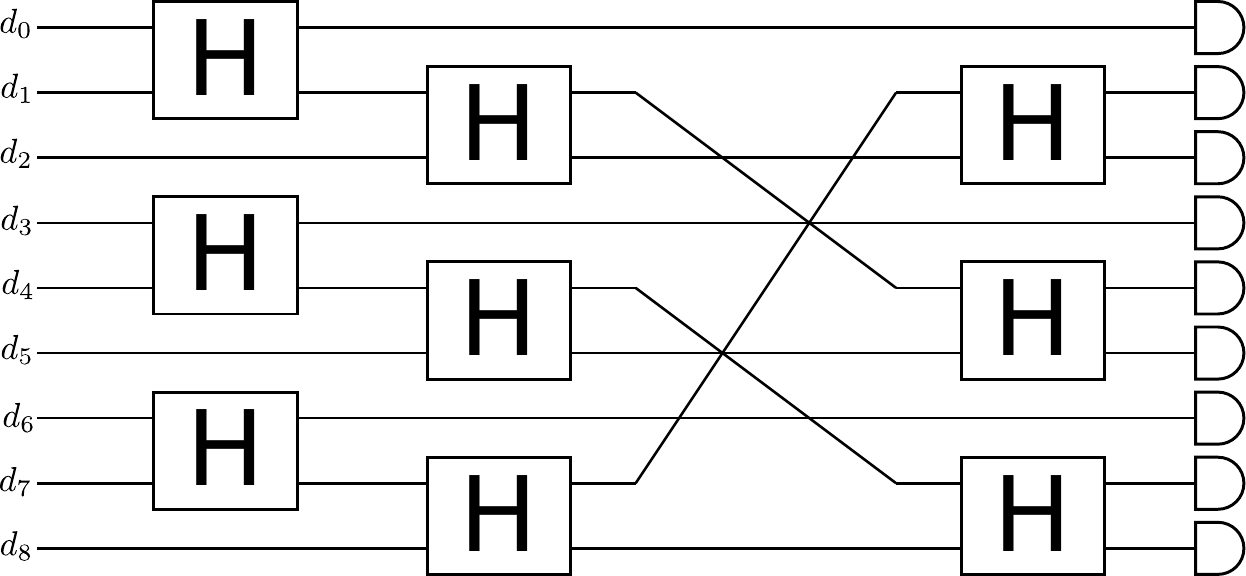}
  \caption{Equivalent circuit for the measurement setup with $H= \frac{1}{\sqrt{2}} \begin{pmatrix}1 & 1 \\ 1 & -1 \end{pmatrix}$. The first row of H matrices corresponds to the passive basis choice. The wires corresponding to $\cZ$ basis modes are directly connected to the detectors. The others go through another row of H matrices for the first beamsplitter of the interferometer, then shift of the temporal modes due to the delay line and another row of H matrices for the second beamsplitter. This circuit is defining $U^{\dagger}$ which is used to construct Bob's POVM. }
  \label{fig:circuit-timebin}
\end{figure}

Finally, we describe Bob's POVM.
We assume here that the threshold detectors are ideal (no dark counts and $100\%$ efficiency), i.e. each detector will click (event denoted by ``C'') if and only if there is one photon or more in the associated mode. The complementary event is labelled ``N''.
The operators describing the threshold detector in mode $j$ are:
\begin{equation}
  \begin{cases}
    \pi_{N_j} =& \proj{\vac}_{d_j}\\
    \pi_{C_j} =& \sum\limits_{n\geq 1} \proj{n}_{d_j}
  \end{cases}
\end{equation}
Then the overall POVM is given by:
\begin{equation}
  \begin{cases}
    \Pi_{N_j} =& U^{\dagger} \pi_{N_j} U\\
    \Pi_{C_j} =& U^{\dagger} \pi_{C_j} U
  \end{cases}
\end{equation}

Here, there are $2^{9}$ possible detection events corresponding to any combination of the detectors clicking or not clicking; we represent them using strings of C and N of length $9$.
We further classify these events by assigning a basis value and an outcome value to each event according to the rules in Table~\ref{table:keymap}.

\begin{table}
  \centering
\begin{tabular}{l|c|c|cl}
\cline{1-3}
\multicolumn{1}{|c|}{\backslashbox{Outcome}{Basis}} & $\cZ$ & $\cX$ &  &  \\ \cline{1-3}
\multicolumn{1}{|c|}{$\varnothing$}       & $(N_0N_1\dots N_8)$ & $(N_0N_1\dots N_8)$  &  &  \\ \cline{1-3}
\multicolumn{1}{|c|}{0}       & $C_0$ & $C_4$                &  &  \\ \cline{1-3}
\multicolumn{1}{|c|}{1}       & $C_3$ & $C_5$                &  &  \\ \cline{1-3}
\multicolumn{1}{|c|}{$\perp$} & $C_6$ & $C_1; C_2; C_7; C_8$ &  &  \\ \cline{1-3}
\multicolumn{1}{|c|}{d} & At least two in $\cZ$ & At least two in $\cX$ &  &  \\ \cline{1-3}
\end{tabular}
\caption{Mapping the detection pattern to basis value and outcome value.
If no detector clicked, then the basis is chosen randomly.
Any event containing at least one click $C_0$, $C_3$ or $C_6$ contributes to a measurement in basis $\cZ$ while any event containing at least one click $C_1$, $C_2$, $C_4$, $C_5$, $C_7$, or $C_8$ contributes to a measurement in basis $\cX$.
If the detection event contains clicks from the two bases, then the conflict is resolved by choosing one basis or the other at random and ignoring all the detector clicks from the other basis. If at least two different clicks from the same basis occurred, then we assign the special symbol ``d'' representing a double click. }
\label{table:keymap}
\end{table}

The outcome $\perp$ represents an event that carries no relevant information; we call it an inconclusive outcome, not to be confused with a no-click outcome $\varnothing$.
For a given input state, it is required that the clicking probabilities in each basis are equal, but the conclusive probabilities in each basis might be different.
This is the case here since the inconclusive probability in basis $\cZ$ is essentially $0$ but in the basis $\cX$, half of the total clicking probability is actually inconclusive due to the interference of only half of the signal in the middle temporal mode.
This can also be understood as a basis dependent trusted erasure channel operating after a basis independent filter.

\subsection{Universal squashing}
\label{ssec:squashing}
In the previous subsection, we modeled the receiver and defined a few relevant outcomes.
Here, we show how to deal with the double clicks to derive the statistics of a virtual single photon protocol whose security will be analysed in the next subsection.

We rely on a generalisation of the universal squashing result first proposed in Ref.~\cite{fung_universal_2011}.
The method proposed in that paper comprises two steps: an equivalence theorem stating that two virtual situations are equivalent and an estimation technique to compute the statistics of one of the two virtual situations using the statistics of the actual protocol.

First, the equivalence theorem (Theorem 1) states that under certain assumptions, the following two situations are equivalent.
In Situation 1, Bob is receiving a $n$ photon single mode input state and only keeps one photon. The single photon evolves through the unitary operation describing the receiver and is measured in only one output arm $j$ using a threshold detector.
In Situation 2, Bob keeps all n photons to interfere through the unitary operation describing the receiver and measures the number of photons in each output arm (let it be $l_j$ for the arm labelled $j$) with a photon-number resolving detector. Bob finally outputs the outcome $j$ with probability $\frac{l_j}{n}$.

Then, the estimation technique tries to estimate the statistics of Situation 2.
The general method is simple: if only one detector clicked in the real protocol, then there were one or more photons in that particular arm and the same outcome would have been announced in Situation 2.
If multiple detectors clicked in the real protocol (double click), then there is no way to resolve the conflict as any outcome could have been announced in Situation 2.
Therefore, the probability of each outcome in Situation 2 is bounded by the statistics of the real protocol: it is lower bounded by the single click probability and upper bounded by the sum of single click and double click probabilities ; the lower the double click rate, the tighter the bounds.

In our paper, we will use the estimation technique exactly as proposed in Ref.~\cite{fung_universal_2011}.

However, our considered application is bringing two issues that we need to consider in order to properly apply the equivalence theorem.
First, the theorem was stated and proved in the single mode case only.
In the protocol under consideration here, Alice has to prepare non-trivial states over several modes (e.g. $\ket{\varphi_2} = \ket{\beta}_{d_0}\ket{\beta}_{d_3}\ket{\vac}_{d_6}$).
Second, the theorem holds for any states having a fixed number of photons.
It is straightforward to generalise to any classical mixture of photon-number states, but here we use coherent states that have coherence between different photon number states and it is not clear if the theorem still holds.

We addressed the first concern with a generalisation of the equivalence theorem in the multi mode case and the second one with an additional argument based on a virtual photon number measurement.
More details are provided in Appendix \ref{app:squashing}.

With this generalised result, we are able to upper and lower bound the statistics of a virtual single-photon protocol after the statistics of the actual implementation using the universal squashing framework.
While the main application here is QKD, the universal squashing framework is a general quantum optics result that could have other applications in single photon Quantum Information Processing. For instance, Ref.~\cite{fung_universal_2011} proposed applications to qubit state tomography.

\subsection{Security analysis}
\label{ssec:winick}
The point of this subsection is to provide a lower bound on the asymptotic key rate of the protocol we described in Subsection \ref{ssec:modeling}.
Here, we restrict the analysis to a qubit protocol since we obtained qubit statistics (or rather upper and lower bounds on them) in the previous subsection \ref{ssec:squashing}.

We choose to use numerical methods to estimate the information leakage to the adversary since they are very practical and often provide better results than analytical methods.
In our case, it is possible to use the method proposed either in Ref.~\cite{winick_reliable_2018} or in Ref.~\cite{moroder_security_2012}.
Both rely on convex optimisation techniques \cite{boyd_convex_2004} and we find that both are giving similar results, the latter being substantially faster though.
The main difference lies in the objective function: Ref.~\cite{winick_reliable_2018} is minimising the quantum relative entropy which is a convex non-linear function while Ref.~\cite{moroder_security_2012} is maximising the phase error rate which is linear.

We consider here the phase error method of Ref.~\cite{moroder_security_2012} for the simulation and we indicate below the main steps to implement it.

We use the entanglement replacement scheme for the transmitter, so instead of preparing the state $\ket{\varphi_i}$ with probability $p_i$, Alice is preparing the bipartite state:
\begin{equation}
  \ket{\psi}_{AA'} = \sum_i \sqrt{p_i} \ket{i}_A \ket{\varphi_i}_{A'}
\end{equation}
She sends the system $A'$ over the channel to Bob and measures the other half $A$ (qutrit) in the computational basis $\set{\proj{i}_A}$.
Alice records her outcome value in a register $\bar{A}$ and her basis choice (announcement) in a classical register $\tilde{A}$ according to Table \ref{table:alice-bob-map}.

Alice's partial state $\rho_A$ is characterised by the overlap of the prepared states and by the preparation probabilities:
\begin{equation}
  \rho_A = \Tr_{A'}\big(\proj{\psi}_{AA'}\big) = \big[\sqrt{p_ip_j}\braket{\varphi_j | \varphi_i}\big]_{ij}
\end{equation}
The information about this state is included in the programme by considering that Alice could perform a tomography of her partial state.
There are 4 Mutually-Unbiased-Bases in dimension 3 and we include the statistics of each of these operators.

For the receiver, while we could in principle write down the whole unitary (in dimension $9$), we choose to simplify it to dimension $3$ instead to speedup the programme.
We consider an active measurement instead, so that only either the direct line or the monitoring line is operating at once on a small dimension system.
As in Ref.~\cite{winick_reliable_2018}, we consider that Bob is manipulating a qutrit where the first two dimensions correspond to a qubit and the third dimension represents the no-click outcome $\varnothing$.
We find that Bob's measurement operators are:
\begin{equation}
  \begin{split}
    \Pi_{\cZ, 0}^B =& \begin{pmatrix}
      \frac{1}{2}&0&0 \\
      0&0&0 \\
      0&0&0 \\
  \end{pmatrix}
  \Pi_{\cZ, 1}^B = \begin{pmatrix}
    0&0&0 \\
    0&\frac{1}{2}&0 \\
    0&0&0 \\
  \end{pmatrix}\\
  \Pi_{\cZ, \perp}^B =& \begin{pmatrix}
  0&0&0 \\
  0&0&0 \\
  0&0&0 \\
  \end{pmatrix}
  \Pi_{\cZ, \varnothing}^B = \begin{pmatrix}
    0&0&0 \\
    0&0&0 \\
    0&0&\frac{1}{2} \\
  \end{pmatrix}\\
  \Pi_{\cX, 0}^B =& \begin{pmatrix}
    \frac{1}{8}&\frac{1}{8}&0 \\
    \frac{1}{8}&\frac{1}{8}&0 \\
    0&0&0 \\
  \end{pmatrix}
  \Pi_{\cX, 1}^B = \begin{pmatrix}
    \frac{1}{8}&-\frac{1}{8}&0 \\
    -\frac{1}{8}&\frac{1}{8}&0 \\
    0&0&0 \\
  \end{pmatrix}\\
  \Pi_{\cX, \perp}^B =& \begin{pmatrix}
  \frac{1}{4}&0&0 \\
  0&\frac{1}{4}&0 \\
  0&0&0 \\
  \end{pmatrix}
  \Pi_{\cX, \varnothing}^B = \begin{pmatrix}
  0&0&0 \\
  0&0&0 \\
  0&0&\frac{1}{2} \\
  \end{pmatrix}
  \end{split}
\end{equation}

After performing a measurement in either of the bases, Bob will register an announcement value in a classical register $\tilde{B}$ and an outcome value in a register $\bar{B}$.
Here the classical register $\tilde{B}$ corresponds to the public announcement Bob will make about his results, but it is not equivalent to a typical basis choice.
The announcement is different for no-click, conclusive (including basis choice) or inconclusive outcomes.
Bob will register his outcome value and announcement value according to Table \ref{table:alice-bob-map}.

\begin{table}
  \begin{tabular}{|c|c|c|}
      \hline
      \backslashbox{$\bar{a}$}{$\tilde{a}$} & 0 & 1\\ \hline
                  0             &   $\ket{\varphi_0}$   &  $\ket{\varphi_2}$ \\ \hline
                  1             &   $\ket{\varphi_1}$   &  Not used \\ \hline
  \end{tabular}
  \begin{tabular}{|c|c|c|c|c|}
      \hline
      \backslashbox{$\bar{b}$}{$\tilde{b}$} & 0 & 1 & 2 & 3\\ \hline
                  0             &   $0_{\cZ}$   &  $0_{\cX}$ &$\varnothing_{\cZ}$& $\perp_{\cZ}$ \\ \hline
                  1             &   $1_{\cZ}$   &  $1_{\cX}$ &$\varnothing_{\cX}$& $\perp_{\cX}$ \\ \hline
  \end{tabular}
  \caption{Announcement and outcome value map for Alice and Bob}
  \label{table:alice-bob-map}
\end{table}
We only consider the announcements $(\tilde{a}=0, \tilde{b}=0)$ for key generation and we use a few more operators to define certain errors and detection probabilities:
\begin{equation}
  \begin{matrix}
    \Pi_{\text{det}, \cZ} =& (\proj{0}_{A}+\proj{1}_{A}) \otimes (\Pi_{\cZ, 0}^B + \Pi_{\cZ, 1}^B) \\
    \Pi_{\text{det}, \cX} =& \proj{2}_{A} \otimes (\Pi_{\cX, 0}^B+\Pi_{\cX, 1}^B) \\
    \Pi_{\text{det, phase}} =& \big(\proj{+}_{A}+\proj{-}_{A}\big) \otimes \big(\Pi_{\cX, 0}^B + \Pi_{\cX, 1}^B\big)\\
    \Pi_{e_{\cZ}} =& \proj{0}_A \otimes \Pi_{\cZ, 1}^B + \proj{1}_A \otimes  \Pi_{\cZ, 0}^B \\
    \Pi_{e_{\cX}} =& \proj{2}_A \otimes \Pi_{\cX, 1}^B \\
    \Pi_{e_{\text{phase}}} =& \proj{+}_{A} \otimes \Pi_{\cX, 1}^B + \proj{-}_{A} \otimes \Pi_{\cX, 0}^B
  \end{matrix}
\end{equation}
where $\ket{\pm}_A = \frac{1}{\sqrt{2}}\big(\ket{0}_A \pm \ket{1}_A\big)$.

Then it is easy to define:
\begin{equation}
  \begin{matrix}
    p_{\text{det}, \cZ} =& \Tr(\rho_{AB}\cdot \Pi_{\text{det}, \cZ}) \\
    p_{\text{det}, \cX} =& \Tr(\rho_{AB}\cdot \Pi_{\text{det}, \cX}) \\
    p_{\text{det, phase}} =& \Tr(\rho_{AB}\cdot \Pi_{\text{det, phase}}) \\
    e_{\cZ} =& \frac{1}{p_{\text{det}, \cZ}}\Tr(\rho_{AB}\cdot \Pi_{e_{\cZ}})\\
    e_{\cX} =& \frac{1}{p_{\text{det}, \cX}}\Tr(\rho_{AB}\cdot \Pi_{e_{\cX}})\\
    e_{\text{phase}} =& \frac{1}{p_{\text{det, phase}}}\Tr(\rho_{AB}\cdot \Pi_{e_{\text{phase}}})
  \end{matrix}
\end{equation}

We note that the bit error in the basis $\cX$ is related to the usual visibility parameter $V$ with the relation
\begin{equation}
  V = 1-2 e_{\cX}
\end{equation}

Finally, the optimisation problem can be cast as a Semi-Definite Programme (SDP):
\begin{equation}
  \begin{matrix}
    \max & e_{\text{phase}}\\
    \text{s. t. } & p_k^{\downarrow} \leq \Tr(\rho_{AB} \Pi_k) \leq p_k^{\uparrow} \\
                  & \rho_{AB} \succeq 0
  \end{matrix}
\end{equation}
for certain measurement operators $\Pi_k$ and lower and upper bounds $p_k^{\downarrow}$ and $p_k^{\uparrow}$ and the key rate is given by:
\begin{equation}
  K \geq p_{\text{det}, \cZ} \big(1 - h_2(e_{\cZ}) - h_2(e_{\text{phase}})\big)
\end{equation}

\section{Results and Discussion}
\label{sec:results}
We perform a simulation to visualise the expected performance of our proposed protocol.
We consider two possible values for the test intensity: $\beta=\alpha$ (same intensity as the key states) or $\beta=\frac{\alpha}{2}$ (one quarter of the intensity of the key state).
We choose a highly biased state preparation where the key states are prepared most of the time: $p_0+p_1=99\%$.
The detectors are assumed to have no dark counts and $100\%$ efficiency.

In Fig.~\ref{fig:keyrate_summary}, we consider a loss-only channel, and in Fig.~\ref{fig:keyrate_cow} we consider a noisy channel with a fixed error rate in both bases $e_{\cZ} = e_{\cX} = 1\%$.

\begin{figure*}
  \resizebox{\textwidth}{!}{\input{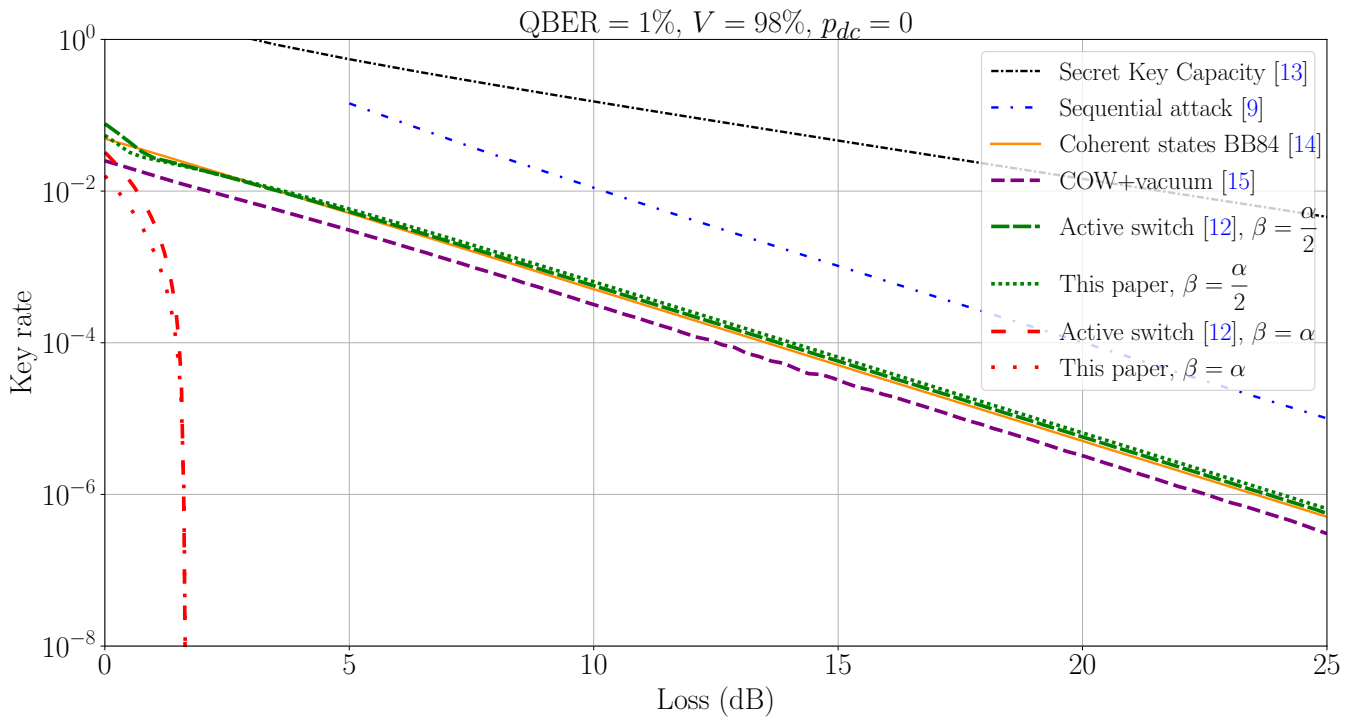}}
  \caption{We simulate a noisy situation with a quantum bit error rate in both bases $e_{\cZ}=e_{\cX} = 1\%$ (or equivalently visibility $98\%$), and ideal threshold detectors with $100\%$ efficiency and no dark counts.
  Two possible modification of the original design offer performances close to another popular design based on phase-encoded coherent states \cite{lo_security_2007}. It is possible either to include an additional test state with vacuum pulses as suggested in Ref.~\cite{curty_foiling_2021} or to modulate the intensity of the test state as mentioned in Ref.~\cite{wang_characterising_2019} to achieve a quadratic scaling. }
  \label{fig:keyrate_cow}
\end{figure*}

Our study reveals that both in the loss-only and in the noisy situation, the original encoding $\beta=\alpha$ cannot guarantee an optimal quadratic scaling, but a proper modulation of $\beta$ can achieve it.
Moreover, any more advanced security analysis involving the inter-phase information would at best only improve the performance marginally since our results (in the loss-only situation) lie close to the upperbound derived in Ref.~\cite{trenyi_zero-error_2021}.

We also analyse one countermeasure proposed in Ref.~\cite{curty_foiling_2021}. If a fourth test state $\ket{\vac}_{c_0}\ket{\vac}_{c_1}\ket{\vac}_{c_2}$ with vacuum pulses only is used along with the original encoding $\beta=\alpha$ for the test state, a quadratic scaling is also achievable.

Surprisingly, those two possible modifications give performances that are similar to a phase-encoded BB84 implementation requiring the preparation of four coherent states with a phase modulation \cite{lo_security_2007}. Thus offering a viable alternative with only limited modifications to the intensity modulator and without phase modulator.

We also notice that our analysis gives results similar to those reported in Ref.~\cite{wang_characterising_2019} with an active switch and a basis independent filter.
We think that the difference in the low loss regime comes from the penalty caused by the use of the universal squashing method: the double click rate is non-negligible in the low loss regime and becomes smaller with a higher channel loss.
Hence it seems that the performance of the protocol is preserved as long as the inconclusive rate (i.e. the clicks in the monitoring line outside of the middle temporal mode; which have been enforced to be zero in Ref.~\cite{wang_characterising_2019} using active switching) corresponds to a trusted erasure channel operating after a basis independent filter.

\section{Conclusion and Outlook}
We presented the security analysis of a coherent state based quantum key distribution protocol against collective attacks in the asymptotic regime.
Our approach relies on the application of the universal squashing framework to bound the single photon statistics, and then the single photon security analysis is performed using numerical methods.
Our simulated results illustrate that our analysis can only establish poor lower bounds on the secure key rate for the original design when not using the inter-signal phase information, but it this bound can also significantly be improved using minor modifications.
Indeed, we have shown that modulating the intensity of the test state $\ket{\beta}\ket{\beta}$ to a lower value $\beta \leq \alpha$, or sending an additional test state $\ket{\text{vac}}\ket{\text{vac}}$ can guarantee that the key rate scales quadratically with the channel loss.
Interestingly, it seems that the performance of this upgraded scheme is comparable to an other popular design based on phase modulated coherent states \cite{lo_security_2007}.

Our results also highlight that the universal squashing framework might have been overlooked when it was initially proposed, while it actually has a considerable interest for applications where it is challenging to obtain an exact squashing model.

The security analysis in the non-asymptotic regime and against general attacks is left to further work.

\section*{Acknowledgements}
We thank Goh Koon Tong and Ignatius W. Primaatmaja for valuable comments and suggestions on the manuscript.
The work is funded by the National Research Foundation (Singapore) Fellowship grant (NRFF11-2019-0001) and Quantum Engineering Programme 1.0 grant (QEP-P2).

\appendix
\section{Universal squashing with multimode coherent states}
\label{app:squashing}
We provide additional details to show that the universal squashing framework is applicable for the security analysis of the COW protocol.

We revisit the equivalence result of Ref.~\cite{fung_universal_2011} (Theorem 1) using a description with optical modes.
They consider a natural squashing operation that keeps only $1$ photon at random out of a pulse of $n$ photons in a single mode.
They show that a protocol implementing this squashing operation has identical statistics as one keeping all $n$ incoming photons, measuring the number $l_j$ of photons in each output arm and annoucing the outcome $j$ with probability $\frac{l_j}{n}$.

We propose a different derivation for their result in the single mode case in Section \ref{ssec:app-single-mode}, and then we generalise it to the multimode case in Section \ref{ssec:app-multi-mode}.
Finally in Section \ref{ssec:app-squashing-coherent} we discuss an additionnal argument to apply the equivalence theorem to coherent states instead of states with a fixed number of photons.

\subsection{Single mode case}
\label{ssec:app-single-mode}
We consider here the case in dimension $2$ to keep equations short but the derivation can be extended to any dimension $d\geq 2$.
We denote the measurement modes $a_0^{\dagger}, a_1^{\dagger}$ and input the modes $b_0^{\dagger}, b_1^{\dagger}$, we assume they are related by a unitary transformation:
\begin{equation}
  \begin{matrix}
    b_0^{\dagger} &= u_{0,0}a_0^{\dagger} + u_{0,1}a_1^{\dagger}\\
    b_1^{\dagger} &= u_{1,0}a_0^{\dagger} + u_{1,1}a_1^{\dagger}\\
  \end{matrix}
\end{equation}

We assume that the input state is a population of $n$ photons in a single mode, say $b_0$ for instance, i.e. the input state is $\frac{1}{\sqrt{n!}}\big(b_0^{\dagger}\big)^n \ket{\vac}$.

After the squashing operation, the single photon is found in arm $j$ with probability $\abs{u_{0,j}}^2$.

Using the other protocol instead, we find the probability of photons in the various arms to be:
\begin{equation}
\label{eq:normalisation_monomode}
\Pr(l_0, l_1) = {n \choose {l_0, l_1}} \abs{u_{0,0}}^{2 l_0}\abs{u_{0,1}}^{2l_1}
\end{equation}

Let us write $X_0 = \abs{u_{0,0}}^2$ and $X_1 = \abs{u_{0,1}}^2$, then the probability of outcome $j$ in this protocol is:
\begin{align}
  \label{eq:cpp_monomode}
    \Pr(j) =& \sum\limits_{l_0+l_1 = n} \frac{l_j}{n}{n \choose {l_0, l_1}} X_0^{l_0} X_1^{l_1}  \\
            =& \frac{1}{n}X_j \frac{\diff \Big[\big( X_0 + X_1 \big)^n \Big] }{\diff X_j} \\
            =& X_j \big( X_0 + X_1\big)^{n-1} \\
            =& X_j = \abs{u_{0,j}}^2
\end{align}

Using this approach, it is easy to see that the photon number probability in each arm as in Eq.~\eqref{eq:normalisation_monomode} is naturally a $0$-th moment (sum to $1$, i.e. normalisation) and the classical postprocessing probability in Eq.~\eqref{eq:cpp_monomode} is a $1$st moment.
We generalise this property to any number of input modes in the next section.

\subsection{Multimode case}
\label{ssec:app-multi-mode}
We show that the equivalence still holds when $d$ modes $b_0^{\dagger}, b_1^{\dagger}, \dots, b_{d-1}^{\dagger}$ are populated with $k_0, k_1, \dots, k_{d-1}$ photons, with $k_0+k_1 + \dots + k_{d-1}=n$.
In this case, the outcome probability for the squashing protocol is:
\begin{equation}
  \label{eq:squashing-stats1}
\Pr(j) = \sum_{i=0}^{d-1} \frac{k_i}{n} \abs{u_{i,j}}^2
\end{equation}

The probability for the second protocol is:
\begin{equation}
  \label{eq:squashing-stats2}
  \begin{split}
    \Pr(j) =& \sum_{l_0+l_1 = n} \frac{l_j}{n} \frac{1}{l_0! l_1! k_0! k_1!}\\
              &\cdot\abs{\braket{\vac | \big(a_0\big)^{l_0}\big(a_1\big)^{l_1} \big(b_0^{\dagger}\big)^{k_0}\big(b_1^{\dagger}\big)^{k_1}|\vac}}^2
  \end{split}
\end{equation}

The equality between Eqs. \ref{eq:squashing-stats1} and \ref{eq:squashing-stats2} is proved below in a slightly more general case (non normalised vectors) and follows three major steps.
First we establish a few results on combinatorics.
Next we show that a particular matrix transformation is a group homomorphism (Claim \ref{eq:app-theorem1}) and show the $0$-th moment property (i.e. normalisation).
Finally, we revisit the derivation of Claim \ref{eq:app-theorem1} to compute the first moment instead, which gives directly the equivalence result in the multimode case.

\subsubsection{Notations}
We start by defining some quantities that are useful to simplify the notations later on.

\begin{definition}[d-multinomial coefficient]
  \begin{equation}
    {n \choose {n_0, n_1, \dots n_{d-1}}} = \begin{cases}\frac{n!}{n_0! \dots n_{d-1}!} & \text{ if } \sum\limits_i n_i = n \\ 0 & \text{ otherwise} \end{cases}
  \end{equation}
\end{definition}

\begin{definition}[Line of total weight $n$]
  A one dimension array $(l_0, \dots, l_{d-1}) \in \set{0 \dots n}^d$ is a line of total weight $n$ if $\sum\limits_i l_i = n$
\end{definition}
We denote $L(n)$ the ensemble of lines of total weight $n$.

$L(n)$ has $N = {n+d-1 \choose {n, d-1}}$ elements, we can index its elements using an index in $\set{0 \dots N-1}$.
We can identify a line with its associated index that we also label $l$ or $l(n)$ if we need to specify the total weight to avoid confusion.
In the following the indices $l, k$ are reserved for elements of $L$, and we use lowerscripts to indicate the component of the solution, e.g. $l_0$ is the $0$-th component of $l(n)$ which is the $l$-th line of total weight $n$ in $L(n)$ and similarly for $k$.

\begin{definition}[Square of total weight $n$]
  A two dimension array $(m_{ij})\in \set{0 \dots n}^{d^2}$ is a square of total weight $n$ if $\sum\limits_{ij}m_{ij}=n$
\end{definition}
We denote $M(n)$ the ensemble of squares of total weight $n$.
If we are given two lines $k(n)$ and $l(n)$, it is also possible to construct more constrained squares where we add the additional constraints:
\begin{align}
  \forall i \in \set{0 \dots d-1}, & \sum\limits_j m_{ij} = k_i \\
  \forall j \in \set{0 \dots d-1}, & \sum\limits_i m_{ij} = l_j
\end{align}

and we denote $M(k, l)$ the ensemble of such squares.

Similarly, we identify a square with its index in $M(n)$ or $M(k,l)$ that we label again $m$ or $m(n)$ or $m(k,l)$ depending on the context. We also occasionally use the following array notation to visualise the sum over rows and columns (here in dimension $2$):
\begin{equation}
  m(k, l)\equiv\begin{tabular}{L L|L}
    m_{00} & m_{01} & k_{0} \\
    m_{10} & m_{11} & k_{1} \\
    \hline
    l_0 & \l_{1} & n
  \end{tabular}
\end{equation}

\begin{definition}[Cubes of total weight $n$]
  A three dimension array $(p_{ijk}) \in \set{0 \dots n}^{d^3}$ is a cube of total weight $n$ if $\sum\limits_{ijk}p_{ijk}=n$
\end{definition}
We label $P(n)$ the ensemble of cubes of total weight $n$.
Again, if we are given three squares $m(k(n), \tilde{l}(n))$, $\bar{m}(k(n), l(n))$ and $\tilde{m}(l(n), \tilde{l}(n))$, we can constraint more the ensemble above by adding the constraints:
\begin{align}
  \sum\limits_{i} p_{ijk} &= \tilde{m}_{kj} \text{  (warning: transpose here !)}\\
  \sum\limits_{j} p_{ijk} &= \bar{m}_{ik}\\
  \sum\limits_{k} p_{ijk} &= m_{ij}
\end{align}
and we identify the elements of $P(m, \tilde{m}, \bar{m})$ with their indices $p(m, \tilde{m}, \bar{m})$ (we can also have only one or two square constraints out of three).
Fig~\ref{fig:cube} represents one element $p \in P(m, \bar{m}, \tilde{m})$.
\begin{center}
  \begin{figure}
    \includegraphics[keepaspectratio, width=\columnwidth]{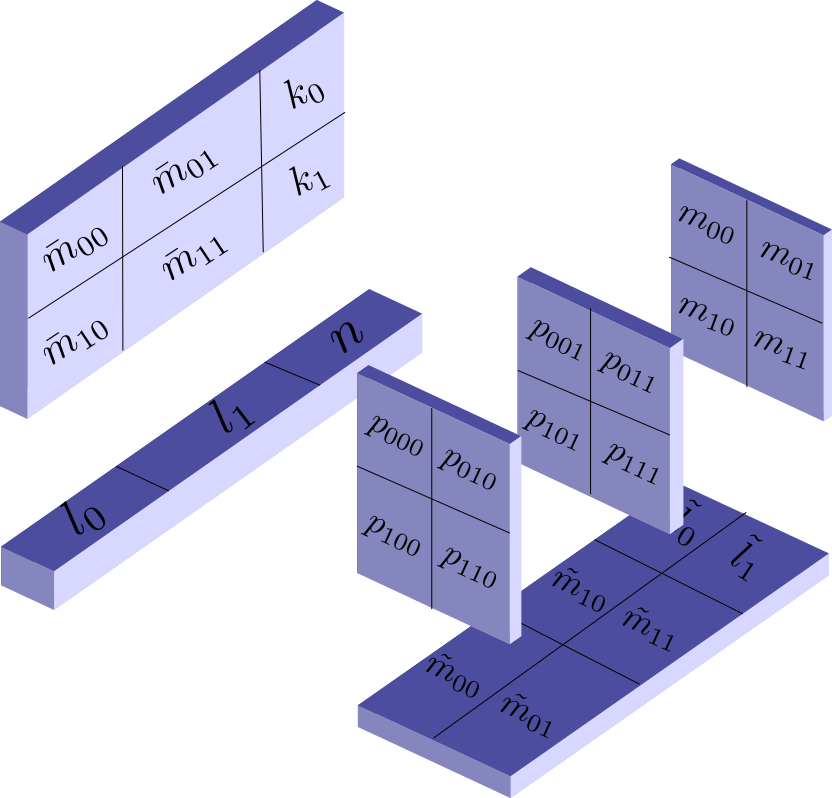}
    \caption{Graphical representation of one $p \in P(m, \bar{m}, \tilde{m})$. The sum over one axis is a projection over on one square. The sum along two axes is a projection on one line. The sum over all three axes is $n$ the total weight of the cube. }
    \label{fig:cube}
  \end{figure}
\end{center}

\subsubsection{Preliminary results}
We introduce here three Lemmas. Lemma 1 is a well know relation and is not directly used to prove our result, but its generalisation will be the main ingredient of Lemma 2.
Lemma 3 is the key ingredient to prove the result in the next subsection and its derivation relies on Lemma 2.

The results are given for $d=2$ to keep equations short, but the results hold for any $d \geq 2$

\begin{lemma}
  \label{eq:lemma1}
  (Vandermonde's identity)
  Let $n, m$ be any two non-negative integer and $l(n) \in L(n)$, then:
  \begin{equation}
    \sum_{k(m)} {l_0 \choose {k_0, l_0-k_0}} {l_1 \choose {k_1, l_1-k_1}} = {n \choose {m, n-m} }
  \end{equation}
\end{lemma}
\begin{proof}
  We consider the expansion of $(1 + X)^n = (1 + X)^{l_0+l_1}$. We find:
  \begin{align}
    (1 + X)^n &= \sum\limits_{m = 0}^n X^m {n \choose {m, n-m}}\\
    (1 + X)^{l_0+l_1} &= \sum\limits_{m=0}^n X^m \sum\limits_{k_0 + k_1 = m} {l_0 \choose {k_0, l_0 - k_1}}{l_1 \choose {k_1, l_1 - k_1}}
  \end{align}
  Then the result follows from the identification of the coefficient in $X^m$.
\end{proof}

\begin{lemma}
  \label{eq:lemma2}
  Let n be any non-negative integer, $k, l \in L(n)$, then:
  \begin{equation}
    \sum\limits_{\begin{tabular}{L L|L}
      m_{00} & m_{01} & k_{0} \\
      m_{10} & m_{11} & k_{1} \\
      \hline
      l_0 & \l_{1} & n
    \end{tabular}} {k_0 \choose {m_{00}, m_{01}}} {k_1 \choose {m_{10}, m_{11}}} = {n \choose {l_0, l_1} }
  \end{equation}
\end{lemma}
\begin{proof}
  We consider here the expansion of $(X_0 + X_1)^n = (X_0 + X_1)^{l_0 + l_1}$
  \begin{align}
    (X_0 + X_1)^n &= \sum\limits_{l(n)} X_0^{l_0}X_1^{l_1} {n \choose {l_0, l_1}}\\
    (X_0 + X_1)^{k_0} &= \sum\limits_{m_{00}+ m_{01} = k_0} X_0^{m_{00}}X_1^{m_{01}} {k_0 \choose {m_{00}, m_{01}}}\\
    (X_0 + X_1)^{k_1} &= \sum\limits_{m_{10}+ m_{11} = k_1} X_0^{m_{10}}X_1^{m_{11}} {k_1 \choose {m_{10}, m_{11}}}
  \end{align}

  Then we take the product and group the monomials:
  \begin{equation}
    \begin{split}
      (X_0 + X_1)^{k_0+k_1} =& \sum\limits_{\substack{m_{00}+ m_{01} = k_0 \\ m_{10}+ m_{11} = k_1}} X_0^{m_{00}+m_{10}}X_1^{m_{01}+m_{11}} \\
      & \cdot{k_0 \choose {m_{00}, m_{01}}}{k_1 \choose {m_{10}, m_{11}}}\\
      =& \sum\limits_{\substack{m_{00}+ m_{10} = l_0 \\ m_{01}+ m_{11} = l_1}} X_0^{l_0}X_1^{l_1} \\
      & \cdot{k_0 \choose {m_{00}, m_{01}}}{k_1 \choose {m_{10}, m_{11}}}\\
      =& \sum\limits_{l(n)} X_0^{l_0}X_1^{l_1} \\
      & \sum\limits_{m(k,l)}{k_0 \choose {m_{00}, m_{01}}}{k_1 \choose {m_{10}, m_{11}}}
    \end{split}
  \end{equation}

  We read out the coefficient in $X_0^{l_0}X_1^{l_1}$ and we find the result.
\end{proof}

\begin{lemma}
  \label{eq:lemma3}
  Let n be any non-negative integer, three lines $k,l,\tilde{l} \in L(n)$ and two squares $\bar{m}(k,l), \tilde{m}(l, \tilde{l})$, then:
  \begin{equation}
    \begin{split}
      \sum\limits_{p(\bar{m}, \tilde{m})} \prod_{ab} \sqrt{{\bar{m}_{ab} \choose {p_{a0b}, p_{a1b}}}{\tilde{m}_{ab} \choose {p_{0ba}, p_{1ba}}}} \\
      = \sqrt{{\prod_a {l_a \choose {\tilde{m}_{a0}, \tilde{m}_{a1}}}}{\prod_b {l_b \choose {\bar{m}_{0b}, \bar{m}_{1b}}}}}
    \end{split}
  \end{equation}
\end{lemma}
\begin{proof}
  We compute the square of the left hand side of Lemma \ref{eq:lemma3}:
  \begin{equation}
    \begin{split}
      \label{eq:lemma3-1}
      \sum\limits_{p(\bar{m}, \tilde{m})} \prod_{ab} \sqrt{{\bar{m}_{ab} \choose {p_{a0b}, p_{a1b}}}{\tilde{m}_{ab} \choose {p_{0ba}, p_{1ba}}}}\\
      \cdot\sum\limits_{\tilde{p}(\bar{m}, \tilde{m})} \prod_{ab} \sqrt{{\bar{m}_{ab} \choose {\tilde{p}_{a0b}, \tilde{p}_{a1b}}}{\tilde{m}_{ab} \choose {\tilde{p}_{0ba}, \tilde{p}_{1ba}}}}\\
      = \sum_{p(\bar{m}, \tilde{m})} \prod_{ab} {\bar{m}_{ab} \choose {p_{a0b}, p_{a1b}}}\\
      \cdot\sum_{\tilde{p}(\bar{m}, \tilde{m})} \prod_{ab} {\tilde{m}_{ab} \choose {\tilde{p}_{0ba}, p_{1ba}}}
    \end{split}
  \end{equation}

  And we can find a closed formula for each independent summation appearing in the right hand side of Eq~\ref{eq:lemma3-1}.

  In the summation over $p$, only two squares are constraining the ensemble: $\bar{m}(k,l)$ and $\tilde{m}(l,\tilde{l})$.
  The idea is to split the cube into $d$ independent squares along the axis of the square that is not constrained (see Fig. \ref{fig:cube}). Then we apply Lemma \ref{eq:lemma2} to each square and we get:
  \begin{equation}
      \sum\limits_{\begin{tabular}{L L|L}
        p_{000} & p_{010} & \bar{m}_{00} \\
        p_{100} & p_{110} & \bar{m}_{10} \\
        \hline
        \tilde{m}_{00} & \tilde{m}_{01} & l_0
      \end{tabular}} {\bar{m}_{00} \choose {p_{000}, p_{010}}}{\bar{m}_{10} \choose {p_{100}, p_{110}}} = {l_0 \choose {\tilde{m}_{00}, \tilde{m}_{01}}}
  \end{equation}
  \begin{equation}
      \sum\limits_{\begin{tabular}{L L|L}
        p_{001} & p_{011} & \bar{m}_{01} \\
        p_{101} & p_{111} & \bar{m}_{11} \\
        \hline
        \tilde{m}_{10} & \tilde{m}_{11} & l_1
      \end{tabular}} {\bar{m}_{01} \choose {p_{001}, p_{011}}}{\bar{m}_{11} \choose {p_{101}, p_{111}}}
      = {l_1 \choose {\tilde{m}_{10}, \tilde{m}_{11}}}
  \end{equation}
  We can multiply these two together to obtain:
  \begin{equation}
    \sum_{p(\bar{m}, \tilde{m})} \prod_{ab} {\bar{m}_{ab} \choose {p_{a0b}, p_{a1b}}} = {\prod_a {l_a \choose {\tilde{m}_{a0}, \tilde{m}_{a1}}}}
  \end{equation}
  In the same way, we find for the other summation:
  \begin{equation}
    \sum_{p(\bar{m}, \tilde{m})} \prod_{ab} {\tilde{m}_{ab} \choose {p_{0ba}, p_{1ba}}} = {\prod_b {l_b \choose {\tilde{m}_{0b}, \tilde{m}_{1b}}}}
  \end{equation}
\end{proof}

\subsubsection{A useful matrix transformation}
We take a fixed non-negative integer $n$ and positive integer $d$.
In the following, we take $d=2$ to have concise equations, but the result is easily generalised to any $d \geq 2$.

Let us consider any square complex matrix $A$ of size $d$.
We denote its elements by $\alpha_{ab}, \forall(a,b)\in \set{0 \dots d-1}^2$.

We construct a complex square matrix $f(A)$ of size $N= \#L(n) = {{n+d-1} \choose {n, d-1}}$ like follows:
\begin{equation}
  \begin{split}
    f(A) =& \sum_{k,l} \ket{k}\bra{l} \sum_{m(k,l)} \sqrt{{k_0 \choose {m_{00}, m_{01}}}{k_1 \choose {m_{10}, m_{11}}}}\\
          & \sqrt{{l_0 \choose {m_{00}, m_{10}}}{l_1 \choose {m_{01}, m_{11}}}} \prod_{ab} \alpha_{ab}^{m_{ab}}
  \end{split}
\end{equation}

\begin{theorem}
  \label{eq:app-theorem1}
  $f$ is preserving the usual matrix multiplication, i.e. for any two square complex matrices $A, B$, we have :
  \[f(A\cdot B) = f(A)\cdot f(B)\]
\end{theorem}

\begin{proof}
Let us take two complex matrices $A, B$ whose elements are respectively $\alpha_{ab}$ and $\beta_{ab}$.
We take $C=A\cdot B$ and we label its elements $\gamma_{ab} = \sum\limits_c \alpha_{ac} \beta_{cb}$.
Then:
\begin{equation}
  \label{eq:th1-1}
\begin{split}
  f(C) &= \sum_{k,\tilde{l}} \ket{k}\bra{\tilde{l}} \sum_{m(k,\tilde{l})} \sqrt{{k_0 \choose {m_{00}, m_{01}}}{k_1 \choose {m_{10}, m_{11}}}}\\
      & \sqrt{{\tilde{l}_0 \choose {m_{00}, m_{10}}}{\tilde{l}_1 \choose {m_{01}, m_{11}}}} \prod_{ab} \Big(\sum\limits_c \alpha_{ac}\beta_{cb}\Big)^{m_{ab}}
\end{split}
\end{equation}

We expand the $d^2$ powers:
\begin{equation}
  \label{eq:th1-powers}
  \begin{split}
    \prod_{ab} \gamma_{ab}^{m_{ab}} =& \prod_{ab}\sum\limits_{p_{ab0}+ p_{ab1} = m_{ab}} \\
    & {m_{ab} \choose {p_{ab0}, p_{ab1}}} \big(\alpha_{a0}\beta_{0b}\big)^{p_{ab0}} \big(\alpha_{a1}\beta_{1b}\big)^{p_{ab1}}
  \end{split}
\end{equation}

In order to regroup the various powers of $\alpha_{ab}$ and $\beta_{ab}$, we denote:
\begin{align}
  \bar{m}_{ab} &= \sum\limits_c p_{acb} \\
  \tilde{m}_{ab} &= \sum\limits_c p_{cba}
\end{align}

We notice that $\bar{m}_{00}+\bar{m}_{10} = \tilde{m}_{00}+\tilde{m}_{01}$, and similarly $\bar{m}_{01}+\bar{m}_{11} = \tilde{m}_{10}+\tilde{m}_{11}$.
We label these two quantities respectively $l_0$ and $l_1$.
We also notice that $\sum\limits_j \bar{m}_{ij} = k_i$ and $\sum\limits_i \tilde{m}_{ij} = \tilde{l}_j$.

We now find:
\begin{equation}
  \label{eq:th1-2}
  \begin{split}
    f(C) =& \sum_{k, \tilde{l}} \ket{k}\bra{\tilde{l}}  \sum_{m(k, \tilde{l})} \sum\limits_{p(m)}
    \sqrt{{k_0 \choose {m_{00}, m_{01}}}{k_1 \choose {m_{10}, m_{11}}}} \\
     & \sqrt{{\tilde{l}_0 \choose {m_{00}, m_{10}}}{\tilde{l}_1 \choose {m_{01}, m_{11}}}}
     \prod\limits_{ab} {m_{ab} \choose {p_{ab0}, p_{ab1}}} \alpha_{ab}^{\bar{m}_{ab}}\beta_{ab}^{\tilde{m}_{ab}}
  \end{split}
\end{equation}

We simplify the $m$ appearing in some multinomial coefficients, and introduce new multinomial coefficients relating $\bar{m}$, $\tilde{m}$ and $l$ to $p(m)$.
Then since $\bar{m}, \tilde{m}, l$ are completely fixed by the definition above, we can add an artificial summation over them that will contain only a single element corresponding to their actual definition depending on $p(m)$.

We find:

\begin{equation}
  \label{eq:th1-3}
  \begin{split}
    f(C) =& \sum\limits_{k, \tilde{l}} \ket{k}\bra{\tilde{l}}  \sum\limits_{m(k, \tilde{l})} \sum\limits_{p(m)} \sum\limits_{l(n)} \sum\limits_{\bar{m}(k,l)}\sum\limits_{\tilde{m}(l, \tilde{l})}
    \prod\limits_{ab}  \alpha_{ab}^{\bar{m}_{ab}}\beta_{ab}^{\tilde{m}_{ab}}\\
          &\sqrt{{k_0 \choose {\bar{m}_{00}, \bar{m}_{01}}} {k_1 \choose {\bar{m}_{10}, \bar{m}_{11}}} {l_0 \choose {\bar{m}_{00}, \bar{m}_{10}}} {l_1 \choose {\bar{m}_{01}, \bar{m}_{11}}}} \\
          &\sqrt{{l_0 \choose {\tilde{m}_{00}, \tilde{m}_{01}}}{l_1 \choose {\tilde{m}_{10}, \tilde{m}_{11}}} {\tilde{l}_0 \choose {m_{00}, m_{10}}} {\tilde{l}_1 \choose {m_{01}, m_{11}}}} \\
          & \sqrt{{\bar{m}_{00}\choose {p_{000}, p_{001}}}{\bar{m}_{01}\choose {p_{100}, p_{101}}}{\bar{m}_{10}\choose {p_{010}, p_{011}}}{\bar{m}_{11}\choose {p_{110}, p_{111}}}}\\
          & \sqrt{{\tilde{m}_{00}\choose {p_{000}, p_{010}}}{\tilde{m}_{01}\choose {p_{001}, p_{011}}}{\tilde{m}_{10}\choose {p_{100}, p_{110}}}{\tilde{m}_{11}\choose {p_{101}, p_{111}}}}\\
          & \sqrt{\bar{m}_{00}!\bar{m}_{01}!\bar{m}_{10}!\bar{m}_{11}!}\sqrt{\tilde{m}_{00}!\tilde{m}_{01}!\tilde{m}_{10}!\tilde{m}_{11}!}\cdot\frac{1}{l_0! l_1!}
  \end{split}
\end{equation}

Next we simplify the summation by shifting the constraints on the summation ensembles to the elements of the summation, swapping the summation order and putting back the constraints in the ensembles.
More precisely, we perform the following operations:
\begin{enumerate}
  \item We introduce $\prod\limits_{ij}\delta_{m_{ij}=\sum\limits_c p_{ijc}}$ in the summation. Here $\delta$ is $1$ if the condition in the lowerscript is satisfied, otherwise it is $0$.
  \item We change the summation over $p(m)$ into one over $p \in P(n)$ (no more constraint in $m(k, \tilde{l})$), the extra elements will make no difference thanks to the deltas we introduced before.
  \item We swap $\sum\limits_{m(k, \tilde{l})}$ and $\sum\limits_p$ since there is no more dependency in $m$ in the latter.
  \item We resolve the summation over $m(k, \tilde{l})$; the deltas introduced before are forcing only one possible $m$ in that sum to make it non-zero, therefore we can remove the whole summation and simplify the deltas.
  \item We push the summation over $p$ last, after the summations over $l, \bar{m}(k,l), \tilde{m}(l,\tilde{l})$.
  \item We put back constraints in the summation over $p$, i.e. the summation is now over $p(\bar{m}, \tilde{m})$
  \item We resolve the summation over $p$ using Lemma \ref{eq:lemma3} and simplify the remaining $\prod\limits_{ij}\sqrt{\bar{m}_{ij}!\tilde{m}_{ij}!}\cdot\frac{1}{l_0!l_1!}$ which is the exact inverse of the quantity obtained with Lemma \ref{eq:lemma3}.
\end{enumerate}

Finally, we have:
\begin{equation}
  \label{eq:th1-4}
  \begin{split}
    f(C) =& \sum\limits_{k, \tilde{l}} \ket{k}\bra{\tilde{l}}  \sum\limits_{l(n)} \sum\limits_{\bar{m}(k,l)}\sum\limits_{\tilde{m}(l, \tilde{l})}
    \prod\limits_{ab}  \alpha_{ab}^{\bar{m}_{ab}}\beta_{ab}^{\tilde{m}_{ab}}\\
          &\sqrt{{k_0 \choose {\bar{m}_{00}, \bar{m}_{01}}} {k_1 \choose {\bar{m}_{10}, \bar{m}_{11}}} {l_0 \choose {\bar{m}_{00}, \bar{m}_{10}}} {l_1 \choose {\bar{m}_{01}, \bar{m}_{11}}}} \\
          &\sqrt{{l_0 \choose {\tilde{m}_{00}, \tilde{m}_{01}}}{l_1 \choose {\tilde{m}_{10}, \tilde{m}_{11}}} {\tilde{l}_0 \choose {m_{00}, m_{10}}} {\tilde{l}_1 \choose {m_{01}, m_{11}}}}
  \end{split}
\end{equation}

This is exactly the contraction $f(A)\cdot f(B)$.
\end{proof}

\begin{corollary}
  \label{eq:corrolary}
  If we denote $\mathcal{U}(d)$ the ensemble of unitary matrices of size $d\geq 2$, then for any fixed $n \geq 1$ and $N={n+d-1 \choose {n, d-1}}$, the mapping $\big(\mathcal{U}(d), \cdot \big) \overset{f}{\longrightarrow} \big(\mathcal{U}(N), \cdot \big )$ is a group homomorphism.
  Hence if $U$ is defining an orthonormal basis, then $f(U)$ is also defining an orthonormal basis.
\end{corollary}
In practice, $f(U)$ is containing the amplitudes of the output state after projection in the measurement basis.
Using Corollary \ref{eq:corrolary}, the basis that we obtain using the transformation $f$ is directly normalised.
We show in the next subsection how to establish the normalisation property in a different way that will be useful to establish the first moment property that we are interested in.

\subsubsection{Computing the 0-th moment}
We consider here a special case where the matrix $A$ is such that $AA^{\dagger}$ is diagonal with $d$ eigenvalues $\lambda_0, \dots, \lambda_{d-1}$.
In other words, we have $\sum\limits_j \abs{\alpha_{ij}}^2 = \lambda_i$.
It is easy to check that $f(A)f(A)^{\dagger}$ is also diagonal with $N$ eigenvalues $\lambda_0 ^{k_0}\dots \lambda_{d-1}^{k_{d-1}}, k \in L(n)$.
Then by reading out the diagonal elements of $f(A)f(A)^{\dagger}$, we find that for all $k(n)$:

\begin{widetext}
  \begin{equation}
    \sum_{l(n)}\left|{f(A)_{kl}}\right|^2
    = \sum_{l(n)} \left|\sum_{m(k,l)}\sqrt{{k_0 \choose {m_{00}, m_{01}}}{k_1 \choose {m_{10}, m_{11}}}{l_0 \choose {m_{00}, m_{10}}}{l_1 \choose {m_{01}, m_{11}}}} \prod_{ab} \alpha_{ab}^{m_{ab}}\right|^2
    = \prod_{i} \big(\sum_{j}\abs{\alpha_{ij}}^2 \big)^{k_i}
  \end{equation}
\end{widetext}

If the matrix $A$ is unitary, we find that it is $1$, i.e. the resulting basis is normalised.

\subsubsection{Computing the 1-st moment}
Now we are interested in computing:
\begin{widetext}
  \begin{equation}
    \label{eq:moment-1}
    \sum_{l(n)}w(l)\abs{f(A)_{kl}}^2
    = \sum_{l(n)} w(l)\left|\sum_{m(k,l)} \sqrt{{k_0 \choose {m_{00}, m_{01}}}{k_1 \choose {m_{10}, m_{11}}}{l_0 \choose {m_{00}, m_{10}}}{l_1 \choose {m_{01}, m_{11}}}} \prod_{ab} \alpha_{ab}^{m_{ab}}\right|^2
  \end{equation}
\end{widetext}

with $w(l) = l_0 \text{ or } l_1$. It should be possible to compute any moment the same way, but we are only concerned about the first moment here.

The idea is to revisit the derivation of Claim \ref{eq:app-theorem1} with this added weight to find a closed formula.
Let us consider the example $w(l) = l_0$.

We rewrite Eq~\eqref{eq:moment-1} into:
\begin{widetext}
  \begin{equation}
  \begin{split}
    \sum_{l(n)}w(l)\left|{f(A)_{kl}}\right|^2 = \sum_{l(n)} w(l) &\sum_{\bar{m}(k,l)}
      \sqrt{{k_0 \choose {\bar{m}_{00}, \bar{m}_{01}}}{k_1 \choose {\bar{m}_{10}, \bar{m}_{11}}}}
      \sqrt{{l_0 \choose {\bar{m}_{00}, \bar{m}_{10}}}{l_1 \choose {\bar{m}_{01}, \bar{m}_{11}}}} \prod_{ab} \alpha_{ab}^{\bar{m}_{ab}} \\
      & \sum_{\tilde{m}(k,l)}\sqrt{{k_0 \choose {\tilde{m}_{00}, \tilde{m}_{01}}}{k_1 \choose {\tilde{m}_{10}, \tilde{m}_{11}}}}
      \sqrt{{l_0 \choose {\tilde{m}_{00}, \tilde{m}_{10}}}{l_1 \choose {\tilde{m}_{01}, \tilde{m}_{11}}}} \prod\limits_{ab} \alpha_{ab}^{*\tilde{m}_{ab}}
    \end{split}
  \end{equation}
\end{widetext}

Then we take $B=A^{\dagger}$ and start from Eq~\eqref{eq:th1-4} (with $k=\tilde{l}$) where we add $w(l)$ after the summation $\sum\limits_{l(n)}$.
We can easily reverse the computation until Eq~\eqref{eq:th1-3}.
Now we change $w(l)$ into $w(p)$ by using the definition of $l$ as a function of $p$.
For our example, we have $w(l) = l_0 = \sum\limits_{ij} p_{ij0}$.
We can further reverse the computation until Eq~\eqref{eq:th1-2}.
At this point, we want to resolve the summation over $p$ with the extra weight $w(p)$.
Originally, we have the following result when there is no extra weight:
\begin{widetext}
  \begin{equation}
    \begin{split}
      \label{eq:moment1-normalisation}
      \prod\limits_{ab}(\alpha_{a0}\beta_{0b}+\alpha_{a1}\beta_{1b})^{m_{ab}} =& \sum\limits_{p(n)}\prod\limits_{ab}{m_{ab} \choose {p_{ab0}, p_{ab1}}} \alpha_{ab}^{\bar{m}_{ab}}\beta_{ab}^{\tilde{m}_{ab}}\\
      =&\sum\limits_{p_{000}+p_{001}=m_{00}} {m_{00} \choose {p_{000}, p_{001}}} \big(\alpha_{00}\beta_{00}\big)^{p_{000}}\big(\alpha_{01}\beta_{10}\big)^{p_{001}}
      \sum\limits_{p_{010}+p_{011}=m_{01}} {m_{01} \choose {p_{010}, p_{011}}} \big(\alpha_{10}\beta_{00}\big)^{p_{010}}\big(\alpha_{11}\beta_{10}\big)^{p_{011}}\\
      &\sum\limits_{p_{100}+p_{101}=m_{10}} {m_{10} \choose {p_{100}, p_{101}}} \big(\alpha_{00}\beta_{01}\big)^{p_{100}}\big(\alpha_{01}\beta_{11}\big)^{p_{101}}
      \sum\limits_{p_{110}+p_{111}=m_{11}} {m_{11} \choose {p_{110}, p_{111}}} \big(\alpha_{10}\beta_{01}\big)^{p_{110}}\big(\alpha_{11}\beta_{11}\big)^{p_{111}}
    \end{split}
  \end{equation}
\end{widetext}

The extra weight will result in a polynomial differentiation with respect to certain parameters.
More precisely, we use the following property like in the monomode case:
\begin{equation}
  \begin{split}
    \sum_{n_0+n_1 = n} {n \choose {n_0, n_1}}n_0 X_0^{n_0}X_1^{n_1} &= X_0 \frac{\diff\big[(X_0 + X_1)^n\big]}{\diff X_0}\\
    &= nX_0\big(X_0 + X_1\big)^{n-1}
  \end{split}
\end{equation}

We split the weight according to the $d^2$ elements $m_{ij}$, for instance we find for $p_{000}$ and then for the full weight $w(p)=l_0$:

\begin{widetext}
  \begin{align}
    \begin{split}
      \sum\limits_{p(m)} p_{000}\prod\limits_{ab} {m_{ab} \choose {p_{ab0}, p_{ab1}}} \alpha_{ab}^{\bar{m}_{ab}}\beta_{ab}^{\tilde{m}_{ab}} =&
      m_{00}\alpha_{00}\beta_{00}
      \big(\alpha_{00}\beta_{00} + \alpha_{01}\beta_{10}\big)^{m_{00}-1}
      \big(\alpha_{00}\beta_{01} + \alpha_{01}\beta_{11}\big)^{m_{01}}\\
       &\cdot\big(\alpha_{10}\beta_{00} + \alpha_{11}\beta_{10}\big)^{m_{10}}
       \big(\alpha_{10}\beta_{01} + \alpha_{11}\beta_{11}\big)^{m_{11}}
    \end{split}\\
    \sum\limits_{p(m)} w(p)\prod\limits_{ab} {m_{ab} \choose {p_{ab0}, p_{ab1}}} \alpha_{ab}^{\bar{m}_{ab}}\beta_{ab}^{\tilde{m}_{ab}} =&
     \prod\limits_{ab}  \gamma_{ab}^{m_{ab}}
     \left(\begin{matrix} &&m_{00}\frac{\alpha_{00}\beta_{00}}{\alpha_{00}\beta_{00} + \alpha_{01}\beta_{10}} &+& m_{01}\frac{\alpha_{00}\beta_{01}}{\alpha_{00}\beta_{01} + \alpha_{01}\beta_{11}}\\
                        &+& m_{10}\frac{\alpha_{10}\beta_{00}}{\alpha_{10}\beta_{00} + \alpha_{11}\beta_{10}} &+& m_{11}\frac{\alpha_{10}\beta_{01}}{\alpha_{10}\beta_{01} + \alpha_{11}\beta_{11}}\end{matrix} \right)
  \end{align}
\end{widetext}

When $A\cdot A^{\dagger}$ is diagonal, we find that necessarily $m_{10}=m_{01}=0$, $m_{00}=k_0$ and $m_{11}=k_1$.
Hence:
\begin{equation}
  \begin{split}
  \sum\limits_{l(n)}l_0\left|{f(A)_{kl}}\right|^2 = & \big(\abs{\alpha_{00}}^2 + \abs{\alpha_{01}}^2\big)^{k_0}\big(\abs{\alpha_{10}}^2 + \abs{\alpha_{11}}^2\big)^{k_1}\\
  &\left( k_0\frac{\abs{\alpha_{00}}^2}{\abs{\alpha_{00}}^2 + \abs{\alpha_{01}}^2} + k_1\frac{\abs{\alpha_{10}}^2}{\abs{\alpha_{10}}^2 + \abs{\alpha_{11}}^2} \right)
\end{split}
\end{equation}

This is equivalent to taking the derivative of the $0-$th moment with respect to the first parameter in each parenthesis.
The same applies for the first moment in $l_1$ which reads:

\begin{equation}
  \begin{array}{ll}
  \sum\limits_{l(n)}l_1\left|{f(A)_{kl}}\right|^2 = & \big(\abs{\alpha_{00}}^2 + \abs{\alpha_{01}}^2\big)^{k_0}\big(\abs{\alpha_{10}}^2 + \abs{\alpha_{11}}^2\big)^{k_1}\\
  &\left( k_0\frac{\abs{\alpha_{01}}^2}{\abs{\alpha_{00}}^2 + \abs{\alpha_{01}}^2} + k_1\frac{\abs{\alpha_{11}}^2}{\abs{\alpha_{10}}^2 + \abs{\alpha_{11}}^2} \right)
\end{array}
\end{equation}

Finally, if we take $A$ to be unitary and further divide the first moment in $l_0$ and $l_1$ by $n$, we recover the universal squashing equivalence for multimodes as in Eq~\eqref{eq:squashing-stats1}.

\subsection{Squashing coherent states}
\label{ssec:app-squashing-coherent}
The generalised result we proved applies for any fixed $n$ and any population of the input modes $k(n)$.
However in practice Bob will never receive an input state with a fixed number of photons since Alice is preparing coherent states that can have an arbitrary number of photons.
It is easy to generalise the result to a classical mixture of photon number states, but here Alice's states are coherent and there is coherence between photon number states of the same mode, hence it is not clear if the result can still apply.

To clarify this point, we use a simple trick.
Since eventually, the detectors are photon-number sensitive, we can always assume that there is a virtual non-destructive measurement of the global photon number across all the modes preceding the actual photon number measurement in each of the arms.
Here it is important to highlight two points.
First this measurement is not performed in practice, but it commutes with the actual measurement, so the statistics will be unchanged and we can always assume it was performed.
Second, this total photon number measurement is a global measurement acting on all $d$ modes at once, and it gives no information about the particular photon-number distribution in each of the arm.

As a result, it is mode basis-independent and coherence-nonbreaking (within the $n$ photon number subspace).
Since the unitary transformation representing the receiver is equivalent to changing basis for the modes, it is equivalent to measure the number of photons before the circuit or after, so we can always assume that the photon number measurement was performed first before the actual receiver transformation.
Now the virtual total photon number measurement will project the input state onto a quantum state (possibly mixed) with a fixed number of photons $n$ and the theorem directly applies.

To see that it is equivalent to measure the total photon number before or after the circuit, let us consider a simple example:
two input modes $a_0, a_1$ and two output modes $b_0, b_1$ are related by a unitary transformation $U$ such that:
\begin{equation}
  \begin{cases}
    b_0 =& u_{00}a_0 + u_{01}a_1\\
    b_1 =& u_{10}a_0 + u_{11}a_1
  \end{cases}
  \label{eq:ex-basis-change}
\end{equation}
Then if we denote $\hat{n}_a = a_0^{\dagger}a_0 + a_1^{\dagger}a_1$ and $\hat{n}_b = b_0^{\dagger}b_0 + b_1^{\dagger}b_1$ the observable of the total number of photons in the two input modes and two output modes respectively, then it is easy to check that $\hat{n}_a = \hat{n}_b$ using Eq.~\eqref{eq:ex-basis-change} and $UU^{\dagger}=U^{\dagger}U=I$.
The general case for any dimension $d$ is derived in a similar way.

\bibliography{references}

\end{document}